\documentclass[11pt,a4paper]{article}

\usepackage{fontspec}
\usepackage[titletoc,title,header]{appendix}

\usepackage{amsmath}
\usepackage{amssymb}
\usepackage{amsthm}
\usepackage{mathtools}
\usepackage{xparse}    % for \NewDocumentCommand

\usepackage[margin=1in]{geometry}

\usepackage{graphicx}
\usepackage{xcolor}

\usepackage{hyperref}
\hypersetup{
  colorlinks=true,
  linkcolor=blue,
  citecolor=blue,
  urlcolor=blue,
  pdfauthor={Yoshitsugu Sekine},
  pdftitle={Phonon Bose-Einstein condensation in a Hubbard-phonon interacting system with infrared divergence}
}

\usepackage[numbers]{natbib}
\theoremstyle{plain}
\newtheorem{thm}{Theorem}[section]
\AtEndEnvironment{thm}{\qed}
\newtheorem{lem}[thm]{Lemma}
\AtEndEnvironment{lem}{\qed}

\AtEndEnvironment{prop}{\qed}

\AtEndEnvironment{cor}{\qed}

\AtEndEnvironment{conj}{\qed}
\newtheorem{fact}[thm]{Fact}
\AtEndEnvironment{fact}{\qed}

\theoremstyle{definition}
\newtheorem{defn}[thm]{Definition}
\AtEndEnvironment{defn}{\qed}

\AtEndEnvironment{ex}{\qed}

\theoremstyle{remark}

\AtEndEnvironment{rem}{\qed}

\makeatletter
\providecommand*{\dashv}{\mathrel{\mathpalette\@Dashv\vDash}}
\newcommand*{\@dashv}[2]{\reflectbox{$\m@th#1#2$}}
\makeatother
\newcommand{\abs}[1]{\left| #1 \right|} % 絶対値
\NewDocumentCommand{\weaknorm}{O{\dbk} m}{#1{#2}} % 弱(準)ノルム
\newcommand{\norm}[1]{\left\Vert #1 \right\Vert} % ノルム
\newcommand{\bkt}[2]{\left\langle #1,\,#2 \right\rangle} % 内積
\newcommand{\rbkt}[2]{\left( #1,\,#2 \right)} % 内積に利用
\DeclareMathOperator*{\slim}{\mathrm{s} \hyphen \lim} % 強極限
\newcommand{\isomto}{\mathrel{\rightarrowtail\kern-1.9ex\twoheadrightarrow}} % 全単射, 同型射, ホモロジー代数・圏論でよく使う
\newcommand{\cbk}[1]{\left\{ #1 \right\}} % 中括弧
\newcommand{\dbk}[1]{\left\langle #1 \right\rangle} % <f>形式の括弧, Dirac bracketの意味でdbk
\newcommand{\pairbk}[1]{\rbk{#1}} % 適当な対のための括弧
\newcommand{\rbk}[1]{\left( #1 \right)} % 丸括弧, 共通
\newcommand{\sqbk}[1]{\left[ #1 \right]} % いわゆる大括弧
\newcommand{\fun}[2]{#1 \rbk{#2}} % 関数テンプレート
\newcommand{\sqfun}[2]{#1 \sqbk{#2}} % 角括弧を使う関数, 確率論の期待値や完全な熱力学関数に利用
\newcommand{\rightopeninterval}[2]{\left[#1, \, #2 \right)} % 右開区間
\NewDocumentCommand{\imunit}{O{\mathsf{i}}}{#1} % 虚数単位, 共通
\NewDocumentCommand{\placeholder}{O{\bullet}}{#1} % 関手やホモロジーなどでハイフンや点で表す要素
\NewDocumentCommand{\trace}{O{\operatorname{Tr}}}{#1} % トレース
\newcommand{\dom}{\operatorname{dom}} % 定義域
\newcommand{\eqcsq}[1]{\sqbk{#1}} % 同値類の記号
\newcommand{\hyphen}{\hbox{-}} % 記号定義の補助
\newcommand{\idone}{1} % 1で表した恒等写像
\newcommand{\laplacian}{\Delta} % ラプラシアン
\newcommand{\napiernum}{\mathsf{e}} % ネイピア数, 自然対数の底, 素電荷など他にもeが出てきて紛らわしい場合に利用
\newcommand{\opchern}[1]{\operatorname{ch}} % チャーン指標
\newcommand{\opimag}{\operatorname{Im}} % 複素数の虚部
\newcommand{\opreal}{\operatorname{Re}} % 複素数の実部
\newcommand{\setSymbolDownLeft}[2]{{\vphantom{#2}}_{#1}{#2}} % 左下に添字をつけるvphantomのエイリアス
\newcommand{\setSymbolUpLeft}[2]{{\vphantom{#2}}^{#1}{#2}} % 左上に添字をつけるvphantomのエイリアス
\NewDocumentCommand{\agvariety}{O{\mathcal}}{#1} % 代数多様体, 特にイデアルが生成する代数多様体
\NewDocumentCommand{\cmdrel}{O{\omega}}{#1} % 分散関係
\NewDocumentCommand{\dfsp}{O{A}}{#1} % 微分形式の空間
\NewDocumentCommand{\eqcpointed}{O{\eqcsq} m}{#1{#2}_{\ast}} % 基点つき同値類
\NewDocumentCommand{\fnheaviside}{O{H}}{#1} % ヘビサイド関数
\NewDocumentCommand{\fthol}{O{\mathcal{O}}}{#1} % 正則関数の空間用
\NewDocumentCommand{\ftmero}{O{\mathcal{M}}}{#1} % 有理型関数の空間用
\NewDocumentCommand{\grcentralizer}{O{Z}}{#1} % 中心化群, expedition0005395
\NewDocumentCommand{\grmetform}{O{2} m m}{\grmet[#1] \! \rbkt{#2}{#3}} % (一般)相対性理論, 計量の二次形式
\NewDocumentCommand{\grmet}{O{2}}{\setSymbolDownLeft{#1}{g}} % (一般)相対性理論, 計量の文字
\NewDocumentCommand{\grnormalizer}{O{N}}{#1} % 正規化群, expedition0005395
\NewDocumentCommand{\gropasym}{O{A}}{#1} % 反対称化作用素, expedition0001365
\NewDocumentCommand{\gropsym}{O{S}}{#1} % 対称化作用素, expedition0001365
\NewDocumentCommand{\grpermorderedpair}{O{\mathcal{P}}}{#1} % 偶数個の要素の異なる順序対への分解の全体, expedition0009453
\NewDocumentCommand{\grsym}{O{\mathfrak{S}} m}{#1_{#2}} % 対称群, expedition0001319
\NewDocumentCommand{\gtbase}{O{\mathcal}}{#1} % 位相の基底, 開集合の基底, expedition0000259
\NewDocumentCommand{\gtfilter}{O{\mathcal}}{#1} % フィルター, expedition0000349
\NewDocumentCommand{\gtfmlclosed}{O{\mathcal}}{#1} % 閉集合系だけではなく, ある集合上の一般の閉集合族を表す, expedition0000227
\NewDocumentCommand{\gtfmlopen}{O{\mathcal}}{#1} % 「位相」の意味での開集合系だけではなく, ある集合上の一般の開集合族を表す, expedition0000227
\NewDocumentCommand{\gtopenball}{O{U}}{#1} % 開球
\NewDocumentCommand{\gtopencover}{O{\mathcal}}{#1} % 開被覆
\NewDocumentCommand{\gtopennbh}{O{\mathcal}}{#1} % 開近傍系
\NewDocumentCommand{\gtpreopencover}{O{\mathcal}}{#1} % 内部を取ると開被覆になる集合の族, expedition0009107
\NewDocumentCommand{\gtsubbase}{O{\mathcal}}{#1} % 位相の準基, expedition0005835
\NewDocumentCommand{\gtvicinity}{O{\mathcal}}{#1} % 一様構造での近縁
\NewDocumentCommand{\lasp}{O{\mathcal}}{#1} % 線型空間
\NewDocumentCommand{\latprightrbk}{O{\top} m}{\rbk{#2}^{#1}} % 転置, expedition0009877
\NewDocumentCommand{\latpright}{O{\top} m}{#2^{#1}} % 転置, expedition0009877
\NewDocumentCommand{\latp}{O{t} m}{\setSymbolUpLeft{#1}{#2}} % 転置, expedition0009877
\NewDocumentCommand{\lpdistribution}{O{\mu} m}{#2_{\ast,#1}} % 分布関数, expedition0010310
\NewDocumentCommand{\lpmollifier}{O{\rho}}{#1} % フリードリクスの軟化子
\NewDocumentCommand{\lpofpositive}{O{\chi}}{#1} % 正型関数
\NewDocumentCommand{\manliederiv}{O{L}}{#1} % リー微分
\NewDocumentCommand{\mansmoothnbh}{O{\mathcal}}{#1} % 多様体上の滑らかな開近傍の族
\NewDocumentCommand{\mblfmldsysgenerated}{O{d} m}{\fun{#1}{#2}} % 生成された$d$-系
\NewDocumentCommand{\mblfmlgenerated}{O{\sigma} m}{\fun{#1}{#2}} % 生成された加法族
\NewDocumentCommand{\oacorrfn}{O{\Gamma}}{#1} % 汎関数積分表示とも関わる相関関数, expedition0011645
\NewDocumentCommand{\oagnsvector}{O{\Omega}}{#1} % GNS表現のベクトル状態
\NewDocumentCommand{\oaideal}{O{\mathcal}}{#1} % 作用素環のイデアル：環や代数のイデアルとは分けておく
\NewDocumentCommand{\oanumberoperator}{O{A}}{#1} % レゾルベント環での個数作用素
\NewDocumentCommand{\oaposcone}{O{\mathcal{P}}}{#1} % 冨田-竹崎理論での正錐
\NewDocumentCommand{\oapressure}{O{P}}{#1} % トレースによる通常の圧力は大文字でP: expedition0009842, トレース状態による修正圧力は小文字でp: expedition0009797
\NewDocumentCommand{\oarepn}{O{\pi}}{#1} % 作用素環の表現
\NewDocumentCommand{\oaspkmsstate}{O{\txtkms}}{#1} % KMS状態の空間
\NewDocumentCommand{\oaspnormalstate}{O{N}}{#1} % 作用素環上の正規状態の空間
\NewDocumentCommand{\oasppurestate}{O{P}}{#1} % 作用素環上の純粋状態の空間
\NewDocumentCommand{\oaspstate}{O{E}}{#1} % 作用素環の状態がなす空間
\NewDocumentCommand{\oastatevector}{O{\Omega}}{#1} % GNS表現に付随する状態ベクトル
\NewDocumentCommand{\oastate}{O{\omega}}{#1} % 作用素環の状態
\NewDocumentCommand{\opdilation}{O{\delta}}{#1} % スケーリング(伸張)の作用素, expedition0010830
\NewDocumentCommand{\opdmat}{O{\rho}}{#1} % 密度行列(密度作用素)
\NewDocumentCommand{\opfockan}{O{a}}{#1} % 場の量子論での消滅作用素
\NewDocumentCommand{\opfockcran}{O{a}}{#1^{\#}} % 場の量子論での生成・消滅作用素をまとめて表す記法
\NewDocumentCommand{\opfockcrdagger}{O{a}}{#1^{\dagger}} % 場の量子論での生成作用素のバージョン
\NewDocumentCommand{\opfockcr}{O{a}}{#1^{\ast}} % 場の量子論での生成作用素
\NewDocumentCommand{\opfocknumber}{O{N}}{#1} % 個数作用素
\NewDocumentCommand{\opfocksegalconj}{O{\pi}}{#1} % 場の量子論でのシーガルの場の作用素, expedition0009962
\NewDocumentCommand{\opfocksegal}{O{\phi}}{#1} % 場の量子論でのシーガルの場の作用素
\NewDocumentCommand{\opfockvac}{O{\Omega}}{#1} % フォック真空
\NewDocumentCommand{\opspecmeas}{O{E}}{#1} % スペクトル測度の標準的な記号
\NewDocumentCommand{\opspec}{O{} m}{\fun{\sigma_{#1}}{#2}} % 作用素のスペクトル, expedition0001551
\NewDocumentCommand{\optransl}{O{\tau}}{#1} % 平行移動の作用素, expedition0010829
\NewDocumentCommand{\opvarspec}{O{} m}{\operatorname{spec}_{#1} #2} % 作用素のスペクトル, スピン系の議論など \sigma ではまぎらわしい場合に利用, expedition0001551
\NewDocumentCommand{\physaction}{O{\mathcal{A}}}{#1} % 物理の作用
\NewDocumentCommand{\physcharge}{O{e}}{\mathrm{#1}} % 電荷, 素電荷などの物理量を表す. e以外にもqに利用する
\NewDocumentCommand{\physcplconst}{O{\mathsf{g}}}{#1} % 結合定数, coupling constant
\NewDocumentCommand{\physelectrostaticcapasity}{O{\mathrm{Cap}}}{#1} % 静電容量, expedition0010917
\NewDocumentCommand{\physenergy}{O{E}}{#1} % エネルギー
\NewDocumentCommand{\physgse}{O{E}}{#1_{0}} % 基底エネルギー, ground state energy
\NewDocumentCommand{\physham}{O{H}}{#1} % ハミルトニアンの雛形, たいていは下つき添字で水素原子, 自由電磁場などを表す
\NewDocumentCommand{\physlagdensity}{O{\mathcal{L}}}{#1} % ラグランジアン密度の雛形
\NewDocumentCommand{\physlag}{O{L}}{#1} % ラグランジアンの雛形
\NewDocumentCommand{\physliouvilean}{O{L}}{#1} % リウビリアンの雛形
\NewDocumentCommand{\physmass}{O{m}}{#1} % 物質の質量
\NewDocumentCommand{\prbbrownmv}{O{B}}{#1} % ブラウン運動
\NewDocumentCommand{\prbcharfun}{O{\chi}}{#1} % 特性関数
\NewDocumentCommand{\prbdist}{O{\mathcal{P}}}{#1} % 確率分布
\NewDocumentCommand{\prbgaussianmeasure}{O{\msrcal{N}}}{#1} % ガウス測度
\NewDocumentCommand{\prbnormaldist}{O{N}}{#1} % 正規分布
\NewDocumentCommand{\prbpoissonprocess}{O{N}}{#1} % ポアソン過程
\NewDocumentCommand{\prbprocess}{O{X}}{#1} % 確率過程
\NewDocumentCommand{\prbqspace}{O{\mathcal{Q}}}{#1} % $Q$-表示の土台の集合を表す
\NewDocumentCommand{\prbspsample}{O{\Omega}}{#1} % 確率空間に対する標本空間, expedition0010114
\NewDocumentCommand{\prbspstate}{O{\mathsf{S}}}{#1} % 確率論での状態空間
\NewDocumentCommand{\psh}{O{\mathfrak}}{#1} % 前層
\NewDocumentCommand{\qtquantumchannel}{O{\mathcal{L}}}{#1} % 量子通信路, expedition0005561, expedition0006844
\NewDocumentCommand{\repn}{O{\pi}}{#1} % 表現
\NewDocumentCommand{\schattencls}{O{\mathbb{K}}}{#1} % コンパクト作用素またはシャッテンクラス
\NewDocumentCommand{\setfmlcylinder}{O{\mathcal{C}}}{#1} % 筒集合族, expedition0000313, expedition0004973
\NewDocumentCommand{\setfml}{O{\mathcal}}{#1} % 集合族
\NewDocumentCommand{\setindex}{O{\mathcal} m}{#1{#2}} % 添字集合
\NewDocumentCommand{\setlattice}{O{\Gamma}}{#1} % 格子, expedition0010278
\NewDocumentCommand{\setspecial}{O{\mathcal} m}{#1{#2}} % 特別な集合を表すための記号, 一点集合・二点集合などを太字で表すといった用途
\NewDocumentCommand{\shdiffform}{O{\sheaf{A}}}{#1} % 微分形式がなす層
\NewDocumentCommand{\sheaf}{O{\mathfrak}}{#1} % 層の記号
\NewDocumentCommand{\smchemicalpotential}{O{\mu}}{#1} % 化学ポテンシャル
\NewDocumentCommand{\smenergydensity}{O{\varrho}}{#1} % エネルギー密度, expedition0011356
\NewDocumentCommand{\smfluctuationwithdmat}{O{\beta} m}{\smuncertaintywithdmat[#1]{#2}^2} % 密度行列に付随する正規状態でのゆらぎ, expedition0011059
\NewDocumentCommand{\sminvtemperature}{O{\beta}}{#1} % 逆温度
\NewDocumentCommand{\smlocaldensityoperator}{O{\rho}}{#1} % 大正準状態に対する局所密度作用素, expedition0011294
\NewDocumentCommand{\smmicrocanonicalstate}{O{\beta} m}{\physmean{#2}_{#1}} % 密度行列に付随する正規状態, expedition0011059
\NewDocumentCommand{\smnumberdensity}{O{\rho}}{#1} % 粒子数密度, expedition0011351
\NewDocumentCommand{\smooth}{O{\mathcal{E}}}{#1} % 滑らかな関数の空間, 特に微分形式関連でも利用, expedition0010649
\NewDocumentCommand{\smparticlenumber}{O{N}}{#1} % 平均粒子数, expedition0011351
\NewDocumentCommand{\smpressure}{O{p}}{#1} % 統計力学での圧力, expedition0011341
\NewDocumentCommand{\smspecificfreeenergy}{O{\bar{f}}}{#1} % (ヘルムホルツの)比自由エネルギー
\NewDocumentCommand{\smthermalvac}{O{\beta}}{\Omega_{#1}} % 熱的真空, expedition0011060
\NewDocumentCommand{\smuncertaintywithdmat}{O{\beta} m}{\rbk{\triangle #2}_{#1}} % 密度行列に付随する正規状態での不確かさ, expedition0011059
\NewDocumentCommand{\sphilbfrak}{O{\mathfrak}}{#1} % ヒルベルト空間
\NewDocumentCommand{\sphilb}{O{\mathcal}}{#1} % ヒルベルト空間
\NewDocumentCommand{\splowerhalf}{O{\mathbb{H}}}{#1_{\txtneg}} % 開下半空間
\NewDocumentCommand{\spupperhalf}{O{\mathbb{H}}}{#1_{\txtnonneg}} % 開上半空間
\NewDocumentCommand{\topmetric}{O{d}}{#1} % 位相空間上の距離関数
\NewDocumentCommand{\vaoutnormal}{O{\widehat}}{#1} % 外向き外法線
\newcommand{\category}[1]{\mathop{\mathsf{#1}}} % 圏を生成する
\newcommand{\catpresheaf}[1]{\category{PSh}} % 前層がなす圏
\newcommand{\faadj}[1]{#1^{\ast}} % 行列・作用素に対する共役・随伴. 線型代数でも内積の必要性からこれに統合. 双対空間にはdualをあてる
\newcommand{\faftr}[1]{\widehat{#1}} % フーリエ変換
\newcommand{\fldcmp}{\fld{C}} % 複素数体
\newcommand{\fldreal}{\fld{R}} % 実数体
\newcommand{\fld}[1]{\mathbb{#1}} % 実数体・複素数体・四元数体, またはこれらをまとめて表す体用のコマンド
\newcommand{\fnexp}[1]{\fun{\exp}{#1}} % 指数関数
\newcommand{\lpseq}{\ell} % 数列空間, expedition0000444
\newcommand{\lp}{L} % ルベーグ空間, L_{\txtloc}^p(X), L^p(X)
\newcommand{\msrcal}[1]{\mathcal{#1}} % mathcalで表す測度
\newcommand{\msrprb}{\mathrm{Pr}} % たまに役立つであろう, 立体のPrで表す確率測度
\newcommand{\oaweyl}{\oa{W}} % 抽象的なワイル環
\newcommand{\oa}[1]{\mathcal{#1}} % 作用素環を表す記号用
\newcommand{\opdmsr}[1]{\mathop{d #1}} % 積分で現れる測度用のコマンド, dxのような単純な対象も含む
\newcommand{\opfocksndqntdiff}{d \Gamma} % 第一種の第二量子化作用素, 微分版
\newcommand{\opfockweyl}{W} % フォック空間上のワイル作用素
\newcommand{\opspbddlin}[1]{\fun{\mathbb{B}}{#1}} % 有界線型作用素がなす空間
\newcommand{\opspecint}[1]{\mathcal{E}} % スペクトル積分, expedition0001621
\newcommand{\physmean}[1]{\dbk{#1}} % 物理, 特に統計力学によく出てくる平均の記号
\newcommand{\prbexp}{\mathbb{E}} % 期待値, 条件つき期待値にも応用
\newcommand{\ringratint}{\mathbb{Z}} % 有理整数環
\newcommand{\seq}[2]{\if\relax\detokenize{#1}\relax \rbk{#1} \else \rbk{#1}_{#2} \fi} % 点列の基本系
\newcommand{\setisomorphism}[1]{\operatorname{Iso}} % 同型写像の集合
\newcommand{\setone}[1]{\cbk{#1}}
\newcommand{\spfockgamma}{\Gamma} % フォック空間
\newcommand{\txtasym}{\mathrm{as}} % 反対称
\newcommand{\txtbec}{\mathrm{BEC}} % BEC, ボース-アインシュタイン凝縮
\newcommand{\txtbsn}{\mathrm{b}} % ボソンのb
\newcommand{\txtcpt}{\mathrm{c}} % コンパクト性を表す添字
\newcommand{\txtcritical}{\mathrm{c}} % 臨界
\newcommand{\txtel}{\mathrm{el}} % 電子のel
\newcommand{\txtfock}{\mathrm{F}} % フォック
\newcommand{\txtfr}{\mathrm{fr}} % 自由場などを表すfr
\newcommand{\txtinteraction}{\mathrm{I}} % 相互作用
\newcommand{\txtircapitalized}{\mathrm{IR}} % 赤外または infrared
\newcommand{\txtkms}{\mathrm{KMS}} % KMS状態の集合を表す
\newcommand{\txtneg}{\mathrm{-}} % 負の部分
\newcommand{\txtnonneg}{\mathrm{+}} % 正または非負の部分
\title{Phonon Bose-Einstein condensation in a Hubbard-phonon interacting system with infrared divergence}

\author{%
Yoshitsugu Sekine\\{\small\texttt{4429sekine@gmail.com}}%
}

\date{\today}

\begin{document}

\maketitle

\begin{abstract}
We study the grand canonical Gibbs states of a finite Hubbard model coupled to massless phonons with the dispersion relation \(\abs{k}^s\), \(1\leq s<d\), at finite temperature. We keep the decoupling setting of the zero-temperature analysis, in which the electron-phonon form factor is common to all sites. An explicit unitary transformation factorizes the Gibbs state into the electron Gibbs state and the free phonon Gibbs state up to a phase, and the phonons exhibit Bose-Einstein condensation in the thermodynamic limit with the critical density of the free Bose gas. Assuming the convergence of this phase as the infrared cutoff is removed, the condensed state is the mixture of mutually disjoint component states with broken gauge symmetry, described by their expectations on the Weyl algebra. Physically, phonons are considered not to undergo Bose-Einstein condensation \cite{VIYukalov001}, and the study of a mechanism that suppresses this mathematical and formal condensation is left for future research.

\end{abstract}

\setcounter{tocdepth}{2}
\tableofcontents

\section{Introduction}\label{introduction}

In \cite{YoshitsuguSekine001} we studied a finite Hubbard model coupled to a massless phonon field at zero temperature. When the electron-phonon form factor is common to all sites, the phonons couple only to the total electron number, and the unitary transformation of Arai and Hirokawa \cite{AraiHirokawa001} decouples the electrons from the phonons. Under the infrared singular condition the ground states then exist only in a representation inequivalent to the Fock representation. In this paper we study the same model at finite temperature. We keep the decoupling setting of \cite{YoshitsuguSekine001}: the form factor is common to all sites throughout. A finite Hubbard system has no phase transition, and the question is whether the phonons condense.

The decoupling extends to a positive chemical potential term (Theorem \ref{thm:hubbard-phonon-bec-corrected-transform}), and the grand canonical Gibbs state in a box is the product of the electron Gibbs state and the free phonon Gibbs state up to a phase (Theorem \ref{thm:hubbard-phonon-equilibrium-factorization}). The analysis of the phonons therefore reduces to the free Bose gas, whose condensation and condensed states were described by Araki and Woods \cite{ArakiWoods001}, Cannon \cite{JohnTCannon001}, and Lewis and Pulè \cite{LewisPule001}; see also \cite[Theorem 5.2.32, pp.~69--70]{BratteliRobinson004} and \cite[Chapters 9 and 10]{AsaoArai020}. For a fixed infrared cutoff the interaction adds a term to the phonon number that does not contribute to the critical density (Theorem \ref{thm:hubbard-phonon-bec-printed-number}), and the phonons condense in the thermodynamic limit above the critical density. Assuming the convergence of the phase as the infrared cutoff is removed, we describe the limit states by their expectations on the Weyl algebra. The condensed state is the mixture of component states with broken gauge symmetry, and the GNS representation is the direct integral of mutually disjoint component representations (Theorem \ref{thm:hubbard-phonon-bec-gns-decomposition}).

Physically, phonons are considered not to undergo Bose-Einstein condensation \cite{VIYukalov001}. The phonon condensation obtained here is a mathematical and formal consequence, and the study of a mechanism that suppresses it is a subject of future research. Equilibrium states and return to equilibrium for small systems coupled to bosonic reservoirs were studied in \cite{BachFrohlichSigal003,DerezinskiJaksicPillet001,JaksicPillet006}, and the absence of condensation is assumed explicitly in \cite{JaksicPillet006}; here the explicit transformation gives the equilibrium states directly.

In Section \ref{sec:hubbard-phonon-bec-settings} we fix the model and derive the factorization. In Section \ref{sec:hubbard-phonon-bec-existence} we prove condensation in the thermodynamic limit with a fixed infrared cutoff. In Section \ref{sec:hubbard-phonon-bec-states} we remove the cutoff and decompose the condensed state.

\section{Model and factorization of the Gibbs state}\label{sec:hubbard-phonon-bec-settings}

In this section we define the model in a box and factorize its Gibbs state.

Let \(\Lambda\) be a finite set, let \(N_{\txtel}\) be an integer with \(1\leq N_{\txtel}\leq2\abs{\Lambda}\), and let \(\sphilbfrak{H}_{\txtel}=\fun{\lpseq^2}{\Lambda;\fldcmp^2}^{\otimes_{\txtasym}N_{\txtel}}\) with the Hubbard Hamiltonian \(H_{\txtel}=\sum_{x,y\in\Lambda}\sum_{s=\pm}t_{x,y}\faadj{c_{x,s}}c_{y,s}+U\sum_xn_{x,+}n_{x,-}\) of \cite{YoshitsuguSekine001}, where \(n_x=n_{x,+}+n_{x,-}\) and \(\sum_xn_x=N_{\txtel}\) on \(\sphilbfrak{H}_{\txtel}\). The phonon one-particle space is \(\sphilbfrak{h}_{\txtbsn}=\fun{\lp^{2}}{\fldreal^{d}}\) with the Fourier transform \(\fun{\faftr{f}}{k}=\int_{\fldreal^d}\fun{f}{x}\napiernum^{-2\pi\imunit x\cdot k}\opdmsr{x}\) and the dispersion relation \begin{equation}\label{eq:hubbard-phonon-bec-dispersion}
\begin{aligned}
\physham[h]_{\txtbsn}=\rbk{-\laplacian}^{s/2},
\quad
\fun{\physham[h]_{\txtbsn}}{k}=\abs{2\pi k}^{s},
\quad
1\leq s<d,
\end{aligned}
\end{equation} where \(\fun{\physham[h]_{\txtbsn}}{k}\) is the multiplier of \(\physham[h]_{\txtbsn}\) in the Fourier representation. The upper bound \(s<d\) is the infrared condition \(\int_{\abs{k}\leq1}\fun{\physham[h]_{\txtbsn}}{k}^{-1}\opdmsr{k}<\infty\). For \(L>0\) let \(C_L=\sqbk{-L/2,L/2}^d\), let \(\sphilbfrak{h}_{\txtbsn,L}=\fun{\lp^{2}}{C_L}\subset\sphilbfrak{h}_{\txtbsn}\), and let \(\physham[h]_{\txtbsn,L}\) be the periodic realization of \(\physham[h]_{\txtbsn}\), with the eigenvectors \(\fun{e_{k,L}}{x}=L^{-d/2}\napiernum^{2\pi\imunit k\cdot x}\) and the eigenvalues \(\fun{\physham[h]_{\txtbsn}}{k}\), \(k\in\setlattice_L^d\), \(\setlattice_L=L^{-1}\ringratint\). The Hilbert space is \(\sphilbfrak{H}_L=\sphilbfrak{H}_{\txtel}\otimes\fun{\spfockgamma_{\txtbsn}}{\sphilbfrak{h}_{\txtbsn,L}}\) with the bosonic Fock space, the field operators, and the second quantization as in \cite{YoshitsuguSekine001,AsaoArai016}. Fix an inverse temperature \(\sminvtemperature>0\), an infrared cutoff \(\kappa>0\), a chemical potential \(\smchemicalpotential<0\), and \(\alpha\in\fldreal\), and put \begin{equation}\label{eq:hubbard-phonon-bec-hamiltonian}
\begin{aligned}
H_{\smchemicalpotential,\kappa,L}
&=
H_{\txtel}\otimes\idone+\idone\otimes H_{\txtbsn,\smchemicalpotential,L}
+\alpha H_{\txtinteraction,\kappa,L},
\\ %%%%%%%%%%%%%%%%
H_{\txtbsn,\smchemicalpotential,L}
&=
\fun{\opfocksndqntdiff_{\txtbsn}}{\physham[h]_{\txtbsn,L}-\smchemicalpotential},
\quad
H_{\txtinteraction,\kappa,L}
=
\sum_{x\in\Lambda}n_x\otimes\fun{\opfocksegal_{\txtfock,\sphilbfrak{h}_{\txtbsn,L}}}{\lambda_{\kappa}}.
\end{aligned}
\end{equation}

\begin{defn}[Assumption]\label{def:hubbard-phonon-bec-common-cutoff}
The form factor $\lambda_{\kappa}\in\sphilbfrak{h}_{\txtbsn,L}$ is common to all
sites, is orthogonal to the zero mode $e_{0,L}$, and belongs to
$\dom\physham[h]_{\txtbsn,L}^{-1}$, where $\physham[h]_{\txtbsn,L}^{-1}$ is the
inverse of the restriction of $\physham[h]_{\txtbsn,L}$ to the orthogonal
complement of $e_{0,L}$; moreover
$\norm{\physham[h]_{\txtbsn,L}^{-1}\lambda_{\kappa}}_{\sphilbfrak{h}_{\txtbsn,L}}$
is bounded in $L$.
\end{defn}

Since \(\napiernum^{-\sminvtemperature\physham[h]_{\txtbsn,L}}\) is of trace class and \(\physham[h]_{\txtbsn,L}-\smchemicalpotential\geq-\smchemicalpotential>0\), \(\napiernum^{-\sminvtemperature H_{\txtbsn,\smchemicalpotential,L}}\) is of trace class \cite[Proposition 5.2.27, p.~57]{BratteliRobinson004}. For a self-adjoint operator \(T\geq c\) on \(\sphilbfrak{h}_{\txtbsn,L}\) with a constant \(c>0\) and for \(g\in\dom T\), \begin{equation}\label{eq:hubbard-phonon-bec-weyl-second-quantization}
\begin{aligned}
\napiernum^{\imunit\fun{\opfocksegal_{\txtfock,\sphilbfrak{h}_{\txtbsn,L}}}{\imunit g}}\fun{\opfocksndqntdiff_{\txtbsn}}{T}\napiernum^{-\imunit\fun{\opfocksegal_{\txtfock,\sphilbfrak{h}_{\txtbsn,L}}}{\imunit g}}
=
\fun{\opfocksndqntdiff_{\txtbsn}}{T}+\fun{\opfocksegal_{\txtfock,\sphilbfrak{h}_{\txtbsn,L}}}{Tg}+\frac{1}{2}\bkt{g}{Tg}_{\sphilbfrak{h}_{\txtbsn,L}}
\end{aligned}
\end{equation} \cite{AsaoArai016}. Define \begin{equation}\label{eq:hubbard-phonon-bec-decoupled-unitary}
\begin{aligned}
V_{\smchemicalpotential,\kappa,L}
&=
\idone\otimes\napiernum^{\imunit\fun{\opfocksegal_{\txtfock,\sphilbfrak{h}_{\txtbsn,L}}}{\imunit\alpha N_{\txtel}\rbk{\physham[h]_{\txtbsn,L}-\smchemicalpotential}^{-1}\lambda_{\kappa}}},
\\ %%%%%%%%%%%%%%%%
\hat{H}_{\txtel,\smchemicalpotential,\kappa,L}
&=
H_{\txtel}-\frac{\alpha^2N_{\txtel}^2}{2}
\bkt{\lambda_{\kappa}}{\rbk{\physham[h]_{\txtbsn,L}-\smchemicalpotential}^{-1}
\lambda_{\kappa}}_{\sphilbfrak{h}_{\txtbsn,L}}.
\end{aligned}
\end{equation} At \(\smchemicalpotential=0\), \(V_{0,\kappa,L}\) is the unitary operator of \cite{YoshitsuguSekine001} in the box.

\begin{thm}\label{thm:hubbard-phonon-bec-corrected-transform}
Under Assumption \ref{def:hubbard-phonon-bec-common-cutoff},
\begin{equation}\label{eq:hubbard-phonon-chemical-transform}
\begin{aligned}
H_{\smchemicalpotential,\kappa,L}
=
V_{\smchemicalpotential,\kappa,L}
\rbk{\hat{H}_{\txtel,\smchemicalpotential,\kappa,L}\otimes\idone
+\idone\otimes H_{\txtbsn,\smchemicalpotential,L}}
V_{\smchemicalpotential,\kappa,L}^{-1}.
\end{aligned}
\end{equation}
\end{thm}

\begin{proof}
Since $\sum_xn_x=N_{\txtel}$,
$\alpha H_{\txtinteraction,\kappa,L}=\idone\otimes\fun{\opfocksegal_{\txtfock,\sphilbfrak{h}_{\txtbsn,L}}}{\alpha N_{\txtel}\lambda_{\kappa}}$.
By \eqref{eq:hubbard-phonon-bec-weyl-second-quantization} with
$T=\physham[h]_{\txtbsn,L}-\smchemicalpotential$ and
$g=\alpha N_{\txtel}\rbk{\physham[h]_{\txtbsn,L}-\smchemicalpotential}^{-1}\lambda_{\kappa}$, $V_{\smchemicalpotential,\kappa,L}\rbk{\idone\otimes H_{\txtbsn,\smchemicalpotential,L}}V_{\smchemicalpotential,\kappa,L}^{-1}$
is $\idone\otimes H_{\txtbsn,\smchemicalpotential,L}+\alpha H_{\txtinteraction,\kappa,L}$
plus the scalar subtracted in \eqref{eq:hubbard-phonon-bec-decoupled-unitary},
and $V_{\smchemicalpotential,\kappa,L}$ commutes with $H_{\txtel}\otimes\idone$.
\end{proof}

Let \(\oa{M}_{\txtel}=\opspbddlin{\sphilbfrak{H}_{\txtel}}\), let \(\oa{M}_{\txtbsn,L}\) be the von Neumann algebra generated by the Weyl operators \(\fun{\opfockweyl_{\txtfock,\sphilbfrak{h}_{\txtbsn,L}}}{f}=\napiernum^{\imunit\fun{\opfocksegal_{\txtfock,\sphilbfrak{h}_{\txtbsn,L}}}{f}}\), \(f\in\sphilbfrak{h}_{\txtbsn,L}\), and \(\oa{M}_L=\oa{M}_{\txtel}\otimes\oa{M}_{\txtbsn,L}\). The Gibbs states are \[\begin{aligned}
\fun{\oastate[\psi_{\sminvtemperature,\smchemicalpotential,\kappa,L}]}{A}
=
\frac{\sqfun{\trace}{A\napiernum^{-\sminvtemperature H_{\smchemicalpotential,\kappa,L}}}}{\sqfun{\trace}{\napiernum^{-\sminvtemperature H_{\smchemicalpotential,\kappa,L}}}},
\quad
\fun{\oastate[\psi_{\txtel,\sminvtemperature}]}{A_{\txtel}}
=
\frac{\sqfun{\trace}{A_{\txtel}\napiernum^{-\sminvtemperature H_{\txtel}}}}{\sqfun{\trace}{\napiernum^{-\sminvtemperature H_{\txtel}}}},
\end{aligned}\] for \(A\in\oa{M}_L\) and \(A_{\txtel}\in\oa{M}_{\txtel}\), and \(\oastate[\psi_{\txtbsn,\txtfr,\sminvtemperature,\smchemicalpotential,L}]\) is the Gibbs state of \(H_{\txtbsn,\smchemicalpotential,L}\) on \(\oa{M}_{\txtbsn,L}\).

\begin{thm}\label{thm:hubbard-phonon-equilibrium-factorization}
Under Assumption \ref{def:hubbard-phonon-bec-common-cutoff},
$\napiernum^{-\sminvtemperature H_{\smchemicalpotential,\kappa,L}}$ is of trace
class, and, for $A_{\txtel}\in\oa{M}_{\txtel}$ and $f\in\sphilbfrak{h}_{\txtbsn,L}$,
\begin{equation}\label{eq:hubbard-phonon-equilibrium-factorization}
\begin{aligned}
&\fun{\oastate[\psi_{\sminvtemperature,\smchemicalpotential,\kappa,L}]}{A_{\txtel}
\otimes \fun{\opfockweyl_{\txtfock,\sphilbfrak{h}_{\txtbsn,L}}}{f}}
\\ %%%%%%%%%%%%%%%%
&=
\napiernum^{-\imunit\alpha N_{\txtel}\opreal\bkt{f}{\rbk{\physham[h]_{\txtbsn,L}
-\smchemicalpotential}^{-1}\lambda_{\kappa}}_{\sphilbfrak{h}_{\txtbsn,L}}}
\fun{\oastate[\psi_{\txtel,\sminvtemperature}]}{A_{\txtel}}
\fun{\oastate[\psi_{\txtbsn,\txtfr,\sminvtemperature,\smchemicalpotential,L}]}{\fun{\opfockweyl_{\txtfock,\sphilbfrak{h}_{\txtbsn,L}}}{f}}.
\end{aligned}
\end{equation}
\end{thm}

\begin{proof}
By \eqref{eq:hubbard-phonon-chemical-transform},
$\napiernum^{-\sminvtemperature H_{\smchemicalpotential,\kappa,L}}$ is unitarily
equivalent to a scalar multiple of
$\napiernum^{-\sminvtemperature H_{\txtel}}\otimes\napiernum^{-\sminvtemperature H_{\txtbsn,\smchemicalpotential,L}}$,
which is of trace class since $\sphilbfrak{H}_{\txtel}$ is finite dimensional.
By the cyclicity of the trace and the Weyl relations
\cite[Proposition 5.2.4(2), p.~13]{BratteliRobinson004},
$V_{\smchemicalpotential,\kappa,L}^{-1}\rbk{A_{\txtel}\otimes\fun{\opfockweyl_{\txtfock,\sphilbfrak{h}_{\txtbsn,L}}}{f}}V_{\smchemicalpotential,\kappa,L}$
is the phase in \eqref{eq:hubbard-phonon-equilibrium-factorization} times
$A_{\txtel}\otimes\fun{\opfockweyl_{\txtfock,\sphilbfrak{h}_{\txtbsn,L}}}{f}$, and the
scalar factor cancels in the normalization.
\end{proof}

\section{Existence of BEC}\label{sec:hubbard-phonon-bec-existence}

In this section we take the thermodynamic limit with a fixed infrared cutoff and show that the phonons condense with the critical density of the free Bose gas.

\begin{thm}\label{thm:hubbard-phonon-bec-printed-number}
Under Assumption \ref{def:hubbard-phonon-bec-common-cutoff}, the expected phonon
number is
\begin{equation}\label{eq:hubbard-phonon-bec-number}
\begin{aligned}
N_{\txtbsn,\kappa,L}
=
\fun{\oastate[\psi_{\sminvtemperature,\smchemicalpotential,\kappa,L}]}{\idone
\otimes\opfocknumber_{\txtbsn,L}}
=
\sum_{k \in \setlattice_L^d}
\frac{1}{\napiernum^{\sminvtemperature
\rbk{\fun{\physham[h]_{\txtbsn}}{k} - \smchemicalpotential}} - 1}
+ \frac{\alpha^2N_{\txtel}^2}{2}
\norm{\rbk{\physham[h]_{\txtbsn,L}-\smchemicalpotential}^{-1}
\lambda_{\kappa}}_{\sphilbfrak{h}_{\txtbsn,L}}^2,
\end{aligned}
\end{equation}
where $\opfocknumber_{\txtbsn,L}$ is the number operator on
$\fun{\spfockgamma_{\txtbsn}}{\sphilbfrak{h}_{\txtbsn,L}}$.
\end{thm}

\begin{proof}
By \eqref{eq:hubbard-phonon-bec-weyl-second-quantization} with the identity operator,
$V_{\smchemicalpotential,\kappa,L}^{-1}\rbk{\idone\otimes\opfocknumber_{\txtbsn,L}}V_{\smchemicalpotential,\kappa,L}$
is $\idone\otimes\opfocknumber_{\txtbsn,L}$ minus a field operator plus the last
term of \eqref{eq:hubbard-phonon-bec-number}.
The free Gibbs state gives the first term and annihilates the field operator by
the gauge invariance \cite[Proposition 5.2.28, p.~60]{BratteliRobinson004}.
\end{proof}

The second term of \eqref{eq:hubbard-phonon-bec-number} comes from the interaction, and its density is controlled by the infrared cutoff through Assumption \ref{def:hubbard-phonon-bec-common-cutoff}. With \(y=\napiernum^{-\sminvtemperature\smchemicalpotential}>1\) put \begin{equation}\label{eq:hubbard-phonon-bec-density}
\begin{aligned}
\fun{\rho_{\txtbsn,\kappa,L}}{\sminvtemperature,y}
&=
\frac{1}{L^d}N_{\txtbsn,\kappa,L}
=
\frac{1}{L^d}\frac{1}{y-1}
+\frac{1}{L^d}\sum_{k\in\setlattice_L^d\setminus\setone{0}}
\frac{1}{y\napiernum^{\sminvtemperature\fun{\physham[h]_{\txtbsn}}{k}}-1}
+\frac{1}{L^d}\fun{N_{\txtbsn,\txtircapitalized,\alpha,\sminvtemperature,\kappa,L}}{y},
\\ %%%%%%%%%%%%%%%%
\fun{N_{\txtbsn,\txtircapitalized,\alpha,\sminvtemperature,\kappa,L}}{y}
&=
\frac{\alpha^2N_{\txtel}^2}{2}
\norm{\rbk{\physham[h]_{\txtbsn,L}-\smchemicalpotential}^{-1}
\lambda_{\kappa}}_{\sphilbfrak{h}_{\txtbsn,L}}^2,
\quad
\fun{\rho_{\txtbsn,\txtircapitalized,\alpha,\kappa,L}}{\sminvtemperature}
=
\frac{1}{L^d}\frac{\alpha^2N_{\txtel}^2}{2}
\norm{\physham[h]_{\txtbsn,L}^{-1}
\lambda_{\kappa}}_{\sphilbfrak{h}_{\txtbsn,L}}^2,
\end{aligned}
\end{equation} so that \(L^{-d}\fun{N_{\txtbsn,\txtircapitalized,\alpha,\sminvtemperature,\kappa,L}}{y}\leq\fun{\rho_{\txtbsn,\txtircapitalized,\alpha,\kappa,L}}{\sminvtemperature}\to0\) as \(L\to\infty\) by Assumption \ref{def:hubbard-phonon-bec-common-cutoff}. The free critical density is \begin{equation}\label{eq:hubbard-phonon-bec-critical-density}
\begin{aligned}
\fun{\rho_{\txtbsn,\txtcritical}}{\sminvtemperature}
=
\int_{\fldreal^{d}}
\frac{1}{\napiernum^{\sminvtemperature\fun{\physham[h]_{\txtbsn}}{k}}-1}\opdmsr{k},
\end{aligned}
\end{equation} which is finite by \(s<d\) and strictly decreasing in \(\sminvtemperature\).

\begin{lem}[{\cite[Lemma 9.22]{AsaoArai020}}]\label{lem:hubbard-phonon-bec-printed-fugacity}
Let $\bar{\rho}_{\txtbsn}>\sup_L\fun{\rho_{\txtbsn,\txtircapitalized,\alpha,\kappa,L}}{\sminvtemperature}$.
For each $L$ there is a unique $y_{\kappa,L}>1$ with
$\fun{\rho_{\txtbsn,\kappa,L}}{\sminvtemperature,y_{\kappa,L}}=\bar{\rho}_{\txtbsn}$,
and there is an increasing sequence $L_n\to\infty$ with
$y_{\kappa,L_n}\to y_{\kappa,\infty}\geq1$.
\end{lem}

\begin{proof}
The function $y\mapsto\fun{\rho_{\txtbsn,\kappa,L}}{\sminvtemperature,y}$ is a
continuous decreasing bijection of $\rbk{1,\infty}$ onto $\rbk{0,\infty}$.
For the sequence, which is not proved correctly in \cite{AsaoArai020},
\eqref{eq:hubbard-phonon-bec-density} gives
$$\begin{aligned}
y_{\kappa,L}-1
\leq
\frac{1}{\bar{\rho}_{\txtbsn}-\fun{\rho_{\txtbsn,\txtircapitalized,\alpha,\kappa,L}}{\sminvtemperature}}
\frac{1}{L^d}\rbk{1+\sum_{k\in\setlattice_L^d}\napiernum^{-\sminvtemperature\fun{\physham[h]_{\txtbsn}}{k}}},
\end{aligned}$$
whose right-hand side is bounded uniformly for sufficiently large $L$ by the
Riemann sum of
$\napiernum^{-\sminvtemperature\fun{\physham[h]_{\txtbsn}}{k}}$.
\end{proof}

The interaction term does not contribute to the critical density for a fixed \(\kappa\), as noted after \eqref{eq:hubbard-phonon-bec-density}. If \(\kappa\to0\) and \(L\to\infty\) simultaneously, or for an infinite Hubbard system, this need not hold. With this observation the arguments of \cite[Lemmas 9.24--9.33, Theorem 9.29]{AsaoArai020} for the free Bose gas apply without change and give the following.

\begin{fact}\label{fact:hubbard-phonon-bec-printed-condensate}
Let $\bar{\rho}_{\txtbsn}$ and $L_n$ be as in
Lemma \ref{lem:hubbard-phonon-bec-printed-fugacity}.
\begin{enumerate}
\item If $\bar{\rho}_{\txtbsn}>\fun{\rho_{\txtbsn,\txtcritical}}{\sminvtemperature}$,
then $y_{\kappa,\infty}=1$ and
$\lim_{n\to\infty}L_n^{-d}\rbk{y_{\kappa,L_n}-1}^{-1}=\bar{\rho}_{\txtbsn}-\fun{\rho_{\txtbsn,\txtcritical}}{\sminvtemperature}$.
\item If $\bar{\rho}_{\txtbsn}\leq\fun{\rho_{\txtbsn,\txtcritical}}{\sminvtemperature}$,
then $\bar{\rho}_{\txtbsn}=\int_{\fldreal^d}\rbk{y_{\kappa,\infty}\napiernum^{\sminvtemperature\fun{\physham[h]_{\txtbsn}}{k}}-1}^{-1}\opdmsr{k}$
and $\lim_{n\to\infty}L_n^{-d}\rbk{y_{\kappa,L_n}-1}^{-1}=0$.
\end{enumerate}
\end{fact}

Let \(\sminvtemperature_{\txtcritical}\) be defined by \(\fun{\rho_{\txtbsn,\txtcritical}}{\sminvtemperature_{\txtcritical}}=\bar{\rho}_{\txtbsn}\). Since \(\fun{\rho_{\txtbsn,\txtcritical}}{\sminvtemperature}\) is strictly decreasing, the zero mode carries the positive number density \(\bar{\rho}_{\txtbsn}-\fun{\rho_{\txtbsn,\txtcritical}}{\sminvtemperature}\) of phonons exactly when \(\sminvtemperature>\sminvtemperature_{\txtcritical}\).

\section{Limit states and their decomposition}\label{sec:hubbard-phonon-bec-states}

In this section we remove the infrared cutoff after the thermodynamic limit and describe the limit states and their decomposition by expectations on the Weyl algebra.

\subsection{The limit states}\label{the-limit-states}

In this subsection we identify the limits of the Gibbs states.

For a complex linear subspace \(\sphilbfrak{D}\subset\sphilbfrak{h}_{\txtbsn}\) let \(\fun{\oaweyl}{\sphilbfrak{D}}\) be the Weyl algebra over \(\sphilbfrak{D}\) with the symplectic form \(\opimag\bkt{f}{g}_{\sphilbfrak{h}_{\txtbsn}}\) and generators \(\fun{\opfockweyl}{f}\) \cite[Theorem 5.2.8, p.~19]{BratteliRobinson004}, and put \begin{equation}\label{eq:hubbard-phonon-bec-test-space}
\begin{aligned}
\sphilbfrak{D}_{\txtbsn,1}
=
\fun{\lp^1}{\fldreal^d} \cap \dom
\rbk{1 - \napiernum^{- \sminvtemperature \physham[h]_{\txtbsn}}}^{-1/2} \cap
\dom \physham[h]_{\txtbsn}^{-1/2}.
\end{aligned}
\end{equation} For \(f\in\sphilbfrak{D}_{\txtbsn,1}\) put \begin{equation}\label{eq:hubbard-phonon-bec-forms}
\begin{aligned}
\fun{q_0}{f}
&=
2\fun{\rho_{\txtbsn,0}}{\sminvtemperature}\abs{\fun{\faftr{f}}{0}}^2,
\quad
\fun{q_1}{f}
=
\bkt{f}{\fun{\coth}{\frac{\sminvtemperature\physham[h]_{\txtbsn}}{2}}f}_{\sphilbfrak{h}_{\txtbsn}},
\\ %%%%%%%%%%%%%%%%
\fun{q_2}{f}
&=
\bkt{f}{\fun{\coth}{\frac{\sminvtemperature\physham[h]_{\txtbsn}+\log y_{\infty}}{2}}f}_{\sphilbfrak{h}_{\txtbsn}},
\end{aligned}
\end{equation} where \(\fun{\rho_{\txtbsn,0}}{\sminvtemperature}=\bar{\rho}_{\txtbsn}-\fun{\rho_{\txtbsn,\txtcritical}}{\sminvtemperature}\) if \(\bar{\rho}_{\txtbsn}>\fun{\rho_{\txtbsn,\txtcritical}}{\sminvtemperature}\), \(y_{\infty}\) is the limit in Fact \ref{fact:hubbard-phonon-bec-printed-condensate}(2), and a pairing \(\bkt{f}{Tf}\) with a nonnegative self-adjoint \(T\) denotes the quadratic form \(\norm{T^{1/2}f}^2\), finite on \(\sphilbfrak{D}_{\txtbsn,1}\).

\begin{lem}[Quasi-free states with constant mean]\label{lem:hubbard-phonon-bec-shifted-states}
Let $y\geq1$ and let $\ell\colon\sphilbfrak{D}_{\txtbsn,1}\to\fldreal$ be real
linear.
Then there is exactly one state
$\oastate[\psi_{\txtbsn,\txtfr,\sminvtemperature,y,\ell}]$ of
$\fun{\oaweyl}{\sphilbfrak{D}_{\txtbsn,1}}$ with
\begin{equation}\label{eq:hubbard-phonon-bec-shifted-states}
\begin{aligned}
\fun{\oastate[\psi_{\txtbsn,\txtfr,\sminvtemperature,y,\ell}]}{\fun{\opfockweyl}{f}}
=
\napiernum^{\imunit\fun{\ell}{f}}
\fnexp{-\frac{1}{4}\bkt{f}{\fun{\coth}{\frac{\sminvtemperature\physham[h]_{\txtbsn}+\log y}{2}}f}_{\sphilbfrak{h}_{\txtbsn}}},
\quad f\in\sphilbfrak{D}_{\txtbsn,1}.
\end{aligned}
\end{equation}
Moreover $\fun{\Theta_{\ell}}{\fun{\opfockweyl}{f}}=\napiernum^{\imunit\fun{\ell}{f}}\fun{\opfockweyl}{f}$
defines a $\ast$-automorphism $\Theta_{\ell}$, and
$\oastate[\psi_{\txtbsn,\txtfr,\sminvtemperature,y,\ell}]=\oastate[\psi_{\txtbsn,\txtfr,\sminvtemperature,y,0}]\circ\Theta_{\ell}$.
\end{lem}

\begin{proof}
Uniqueness holds since the linear span of the generators is a dense
$\ast$-subalgebra.
The unitaries $\napiernum^{\imunit\fun{m}{f}}\fun{\opfockweyl}{f}$ satisfy the Weyl
relations for every real linear $m$, so that
\cite[Theorem 5.2.8, p.~19]{BratteliRobinson004} gives the automorphism
$\Theta_m$.
Let $\oastate[\psi_{\txtfock}]$ be the vacuum state of the Fock representation,
with value $\napiernum^{-\norm{f}^2/4}$ at $\fun{\opfockweyl}{f}$
\cite{AsaoArai016}.
The Planck density satisfies
$1+2\rbk{y\napiernum^{\sminvtemperature\physham[h]_{\txtbsn}}-1}^{-1}=\fun{\coth}{\rbk{\sminvtemperature\physham[h]_{\txtbsn}+\log y}/2}$,
and the domain of its square root contains $\sphilbfrak{D}_{\txtbsn,1}$.
For $X$ in the span of the generators $\fun{\opfockweyl}{f}$ with $f$ in a
finite-dimensional real subspace $E$, put
$\fun{\oastate[\psi_{\txtbsn,\txtfr,\sminvtemperature,y,0}]}{X}=\int\fun{\oastate[\psi_{\txtfock}]}{\fun{\Theta_m}{X}}\opdmsr{\fun{\msrprb_E}{m}}$,
where each $m$ in the dual of $E$ is extended real linearly to
$\sphilbfrak{D}_{\txtbsn,1}$, $\fun{\Theta_m}{X}$ does not depend on the extension,
and $\msrprb_E$ is the centered Gaussian measure on the dual of $E$ with
covariance $\opreal\bkt{\rbk{y\napiernum^{\sminvtemperature\physham[h]_{\txtbsn}}-1}^{-1/2}f}{\rbk{y\napiernum^{\sminvtemperature\physham[h]_{\txtbsn}}-1}^{-1/2}g}$.
This is independent of $E$ by the consistency of these Gaussian measures, positive
and bounded by $\norm{X}$, and hence extends to a state with the values
\eqref{eq:hubbard-phonon-bec-shifted-states} for $\ell=0$.
Composition with $\Theta_{\ell}$ gives the general case.
\end{proof}

For \(\kappa>0\) and \(L\) let \(\smchemicalpotential_{\kappa,L}=-\sminvtemperature^{-1}\log y_{\kappa,L}\) with \(y_{\kappa,L}\) of Lemma \ref{lem:hubbard-phonon-bec-printed-fugacity}. Assume that the limit \[\begin{aligned}
\fun{\tilde{n}_{\txtel}}{f}
=
\lim_{\kappa\downarrow0}\lim_{n\to\infty}
N_{\txtel}\opreal\bkt{f}{\rbk{\physham[h]_{\txtbsn,L_n}-\smchemicalpotential_{\kappa,L_n}}^{-1}\lambda_{\kappa}}_{\sphilbfrak{h}_{\txtbsn,L_n}}
\end{aligned}\] exists for each test function \(f\). By Theorem \ref{thm:hubbard-phonon-equilibrium-factorization}, the phonon factor of the Gibbs state is the free phonon Gibbs state \(\oastate[\psi_{\txtbsn,\txtfr,\sminvtemperature,\smchemicalpotential_{\kappa,L_n},L_n}]\), whose limits are those of the free Bose gas \cite[Lemma 10.8, Theorem 10.9]{AsaoArai020}; for \(s=2\) with Dirichlet boundary conditions the analogous limits are \cite[Theorem 5.2.32, pp.~69--70]{BratteliRobinson004}.

\begin{fact}\label{fact:hubbard-phonon-bec-free-limits}
For $\kappa>0$ and
$f\in\fun{C_{\txtcpt}}{\fldreal^d}\cap\sphilbfrak{D}_{\txtbsn,1}$,
$$\begin{aligned}
\lim_{n\to\infty}
\fun{\oastate[\psi_{\txtbsn,\txtfr,\sminvtemperature,\smchemicalpotential_{\kappa,L_n},L_n}]}{\fun{\opfockweyl_{\txtfock,\sphilbfrak{h}_{\txtbsn,L_n}}}{f}}
=
\napiernum^{-\frac{1}{4}\rbk{\fun{q_0}{f}+\fun{q_1}{f}}}
\end{aligned}$$
if $\bar{\rho}_{\txtbsn}>\fun{\rho_{\txtbsn,\txtcritical}}{\sminvtemperature}$;
the right-hand side is $\napiernum^{-\fun{q_2}{f}/4}$ if
$\bar{\rho}_{\txtbsn}<\fun{\rho_{\txtbsn,\txtcritical}}{\sminvtemperature}$ and
$\napiernum^{-\fun{q_1}{f}/4}$ if
$\bar{\rho}_{\txtbsn}=\fun{\rho_{\txtbsn,\txtcritical}}{\sminvtemperature}$.
\end{fact}

\begin{thm}\label{thm:hubbard-phonon-bec-limit-states}
For $A_{\txtel}\in\oa{M}_{\txtel}$ and
$f\in\fun{C_{\txtcpt}}{\fldreal^d}\cap\sphilbfrak{D}_{\txtbsn,1}$,
$$\begin{aligned}
&\lim_{\kappa\downarrow0}\lim_{n\to\infty}
\fun{\oastate[\psi_{\sminvtemperature,\smchemicalpotential_{\kappa,L_n},\kappa,L_n}]}{A_{\txtel}
\otimes\fun{\opfockweyl_{\txtfock,\sphilbfrak{h}_{\txtbsn,L_n}}}{f}}
\\ %%%%%%%%%%%%%%%%
&=
\napiernum^{-\imunit\alpha\fun{\tilde{n}_{\txtel}}{f}}
\fun{\oastate[\psi_{\txtel,\sminvtemperature}]}{A_{\txtel}}
\lim_{n\to\infty}
\fun{\oastate[\psi_{\txtbsn,\txtfr,\sminvtemperature,\smchemicalpotential_{\kappa,L_n},L_n}]}{\fun{\opfockweyl_{\txtfock,\sphilbfrak{h}_{\txtbsn,L_n}}}{f}},
\end{aligned}$$
where the last limit is given by Fact \ref{fact:hubbard-phonon-bec-free-limits}
and does not depend on $\kappa$.
\end{thm}

\begin{proof}
For $n$ so large that the support of $f$ lies in $C_{L_n}$,
\eqref{eq:hubbard-phonon-equilibrium-factorization} writes the expectation as
the phase
$\napiernum^{-\imunit\alpha N_{\txtel}\opreal\bkt{f}{\rbk{\physham[h]_{\txtbsn,L_n}-\smchemicalpotential_{\kappa,L_n}}^{-1}\lambda_{\kappa}}_{\sphilbfrak{h}_{\txtbsn,L_n}}}$
times the electron factor, which depends neither on $\kappa$ nor on $n$, times the
free phonon factor.
The phonon factor converges as $n\to\infty$ by
Fact \ref{fact:hubbard-phonon-bec-free-limits}, and the phase converges to
$\napiernum^{-\imunit\alpha\fun{\tilde{n}_{\txtel}}{f}}$ as $n\to\infty$ and then
$\kappa\downarrow0$ by the assumption above.
\end{proof}

By Lemma \ref{lem:hubbard-phonon-bec-shifted-states}, the limit in Fact \ref{fact:hubbard-phonon-bec-free-limits} is the value at \(\fun{\opfockweyl}{f}\) of \(\oastate[\psi_{\txtbsn,\txtfr,\sminvtemperature,y_{\infty},0}]\) for \(\bar{\rho}_{\txtbsn}<\fun{\rho_{\txtbsn,\txtcritical}}{\sminvtemperature}\) and of \(\oastate[\psi_{\txtbsn,\txtfr,\sminvtemperature,1,0}]\) at the critical density, and the phase is given by \(\Theta_{\ell}\) with \(\fun{\ell}{f}=-\alpha\fun{\tilde{n}_{\txtel}}{f}\). The condensed case is identified in Definition \ref{def:hubbard-phonon-bec-component-states} below.

\subsection{Decomposition of the condensed state}\label{decomposition-of-the-condensed-state}

In this subsection we assume \(\fun{\rho_{\txtbsn,0}}{\sminvtemperature}>0\) and decompose the condensed phonon state into components with broken gauge symmetry.

For \(\pairbk{r,\theta}\in\rightopeninterval{0}{\infty}\times\rightopeninterval{0}{2\pi}\) put \begin{equation}\label{eq:hubbard-phonon-bec-mean-functional}
\begin{aligned}
\fun{\ell_{\sminvtemperature,r,\theta}}{f}
=
\sqrt{2\fun{\rho_{\txtbsn,0}}{\sminvtemperature}r}\,
\fun{\opreal}{\napiernum^{\imunit\theta}\fun{\faftr{f}}{0}},
\quad
f\in\sphilbfrak{D}_{\txtbsn,1},
\end{aligned}
\end{equation} and let \(\msrprb\) be the probability measure \(\opdmsr{\msrprb(r,\theta)}=\napiernum^{-r}\opdmsr{r}\frac{\opdmsr{\theta}}{2\pi}\). If \(X_1+\imunit X_2=\sqrt{r}\napiernum^{\imunit\theta}\), then \(X_1\) and \(X_2\) are independent centered Gaussian variables of variance \(1/2\) under \(\msrprb\), and \(\ell_{\sminvtemperature,r,\theta}(f)\) is centered Gaussian of variance \(\fun{q_0}{f}/2\); hence \begin{equation}\label{eq:hubbard-phonon-bec-polar-representation}
\begin{aligned}
\fnexp{-\frac{1}{4}\fun{q_0}{f}}
=
\sqfun{\prbexp_{\msrprb(r,\theta)}}{\napiernum^{\imunit\fun{\ell_{\sminvtemperature,r,\theta}}{f}}},
\quad
f\in\sphilbfrak{D}_{\txtbsn,1}.
\end{aligned}
\end{equation}

\begin{defn}[Component states and BEC state]\label{def:hubbard-phonon-bec-component-states}
The component state with order coordinate $\pairbk{r,\theta}$ is
$\oastate[\psi_{\txtbsn,\txtfr,\sminvtemperature,1,\ell_{\sminvtemperature,r,\theta}}]$
of Lemma \ref{lem:hubbard-phonon-bec-shifted-states}.
The BEC state is
\begin{equation}\label{eq:hubbard-phonon-bec-state-mixture}
\begin{aligned}
\fun{\oastate[\psi_{\txtbsn,\txtfr,\txtbec,\sminvtemperature}]}{X}
&=
\sqfun{\prbexp_{\msrprb(r,\theta)}}{\fun{\oastate[\psi_{\txtbsn,\txtfr,\sminvtemperature,1,\ell_{\sminvtemperature,r,\theta}}]}{X}},
\quad X\in\fun{\oaweyl}{\sphilbfrak{D}_{\txtbsn,1}}.
\end{aligned}
\end{equation}
\end{defn}

The integrand is continuous in \(\pairbk{r,\theta}\), and, by \eqref{eq:hubbard-phonon-bec-polar-representation}, \(\fun{\oastate[\psi_{\txtbsn,\txtfr,\txtbec,\sminvtemperature}]}{\fun{\opfockweyl}{f}}=\napiernum^{-\rbk{\fun{q_0}{f}+\fun{q_1}{f}}/4}\), the phonon factor of Theorem \ref{thm:hubbard-phonon-bec-limit-states} in the condensed case. This mixture is the analogue for our dispersion of the decomposition of the condensed free Bose gas \cite{LewisPule001}. The gauge transformation \(f\mapsto\napiernum^{\imunit\vartheta}f\) maps the component with phase \(\theta\) to that with phase \(\theta+\vartheta\) and leaves the BEC state invariant. Moreover, for \(t\in\fldreal\) and \(f\) with \(\napiernum^{\imunit t\physham[h]_{\txtbsn}}f\in\sphilbfrak{D}_{\txtbsn,1}\), the BEC state takes the same value at \(\fun{\opfockweyl}{\napiernum^{\imunit t\physham[h]_{\txtbsn}}f}\) and at \(\fun{\opfockweyl}{f}\), since \(\fun{\physham[h]_{\txtbsn}}{0}=0\).

\begin{lem}[Zero-mode test vectors]\label{lem:hubbard-phonon-bec-zero-mode-vectors}
For every $z\in\fldcmp$ there are $f_n\in\sphilbfrak{D}_{\txtbsn,1}$ with
$\fun{\faftr{f_n}}{0}=z$, $\norm{f_n}\to0$, and $\fun{q_1}{f_n}\to0$.
\end{lem}

\begin{proof}
Take $g\in\fun{C_{\txtcpt}^{\infty}}{\fldreal^d}$ with $\fun{\faftr{g}}{0}=z$ and
$\fun{f_n}{x}=n^{-d}\fun{g}{x/n}$, so that $\fun{\faftr{f_n}}{k}=\fun{\faftr{g}}{nk}$
and $\norm{f_n}^2=n^{-d}\norm{g}^2$.
Since $\fun{\coth}{a/2}\leq1+2/a$ and $\rbk{1-\napiernum^{-a}}^{-1}\leq1+1/a$, it
suffices that
$\int\abs{\fun{\faftr{g}}{nk}}^2\fun{\physham[h]_{\txtbsn}}{k}^{-1}\opdmsr{k}\to0$.
On $\abs{k}\geq1$ the integral is at most $n^{-d}\rbk{2\pi}^{-s}\norm{g}^2$, and on
$\abs{k}\leq1$ it tends to $0$ by dominated convergence since $s<d$.
\end{proof}

Fix a GNS triple \(\pairbk{\sphilbfrak{H}_{\txtbsn,\sminvtemperature},\oarepn_{\txtbsn,\sminvtemperature},\oagnsvector[\Psi_{\txtbsn,\sminvtemperature}]}\) of \(\oastate[\psi_{\txtbsn,\txtfr,\sminvtemperature,1,0}]\).

\begin{thm}\label{thm:hubbard-phonon-bec-gns-decomposition}
Assume $\fun{\rho_{\txtbsn,0}}{\sminvtemperature}>0$.
\begin{enumerate}
\item For every $\pairbk{r,\theta}$,
$\pairbk{\sphilbfrak{H}_{\txtbsn,\sminvtemperature},\oarepn_{\txtbsn,\sminvtemperature}\circ\Theta_{\ell_{\sminvtemperature,r,\theta}},\oagnsvector[\Psi_{\txtbsn,\sminvtemperature}]}$
is a GNS triple of the component state.
\item If $\sqrt{r}\napiernum^{\imunit\theta}\neq\sqrt{r'}\napiernum^{\imunit\theta'}$,
the only bounded intertwiner of the component representations with order
coordinates $\pairbk{r,\theta}$ and $\pairbk{r',\theta'}$ is $0$.
\item The representation
$$\begin{aligned}
\fun{\rbk{\fun{\oarepn_{\txtbsn,\txtbec,\sminvtemperature}}{X}\Psi}}{r,\theta}
=\fun{\oarepn_{\txtbsn,\sminvtemperature}}{\fun{\Theta_{\ell_{\sminvtemperature,r,\theta}}}{X}}\fun{\Psi}{r,\theta}
\end{aligned}$$
on $\fun{\lp^2}{\rightopeninterval{0}{\infty}\times\rightopeninterval{0}{2\pi},\msrprb;\sphilbfrak{H}_{\txtbsn,\sminvtemperature}}$
and the constant section with value $\oagnsvector[\Psi_{\txtbsn,\sminvtemperature}]$
form a GNS triple of the BEC state.
\end{enumerate}
\end{thm}

\begin{proof}
Statement (1) follows from Lemma \ref{lem:hubbard-phonon-bec-shifted-states} and
the surjectivity of $\Theta_{\ell_{\sminvtemperature,r,\theta}}$.

We next show a strong limit that detects the order coordinate.
For $z\in\fldcmp$ and $f_n$ of Lemma \ref{lem:hubbard-phonon-bec-zero-mode-vectors},
$\norm{\rbk{\fun{\oarepn_{\txtbsn,\sminvtemperature}}{\fun{\opfockweyl}{f_n}}-\idone}\oagnsvector[\Psi_{\txtbsn,\sminvtemperature}]}^2=2-2\napiernum^{-\fun{q_1}{f_n}/4}\to0$,
and the Weyl relations with $\abs{\opimag\bkt{f_n}{g}}\leq\norm{f_n}\norm{g}\to0$
extend this to the total set of vectors
$\fun{\oarepn_{\txtbsn,\sminvtemperature}}{\fun{\opfockweyl}{g}}\oagnsvector[\Psi_{\txtbsn,\sminvtemperature}]$.
Since $\fun{\ell_{\sminvtemperature,r,\theta}}{f_n}$ does not depend on $n$,
\begin{equation}\label{eq:hubbard-phonon-bec-central-limit}
\begin{aligned}
\slim_{n\to\infty}
\fun{\oarepn_{\txtbsn,\sminvtemperature}}{\fun{\Theta_{\ell_{\sminvtemperature,r,\theta}}}{\fun{\opfockweyl}{f_n}}}
=
\fnexp{\imunit\sqrt{2\fun{\rho_{\txtbsn,0}}{\sminvtemperature}r}\,
\fun{\opreal}{\napiernum^{\imunit\theta}z}}.
\end{aligned}
\end{equation}
For (2), an intertwiner $Y$ satisfies
$\napiernum^{\imunit\sqrt{2\fun{\rho_{\txtbsn,0}}{\sminvtemperature} r}\opreal\rbk{\napiernum^{\imunit\theta}z}}Y=\napiernum^{\imunit\sqrt{2\fun{\rho_{\txtbsn,0}}{\sminvtemperature} r'}\opreal\rbk{\napiernum^{\imunit\theta'}z}}Y$
for all $z$, where $\fun{\rho_{\txtbsn,0}}{\sminvtemperature}>0$.
Since $\sqrt{r}\napiernum^{\imunit\theta}\neq\sqrt{r'}\napiernum^{\imunit\theta'}$,
the two real linear functions of $z$ in the exponents differ at some $z_0$, and
for $z=tz_0$ with a small real $t\neq0$ the two phase factors are distinct, so
that $Y=0$.

For (3), the operators $\fun{\Theta_{\ell_{\sminvtemperature,r,\theta}}}{X}$ depend
norm continuously on $\pairbk{r,\theta}$, the vector state at the constant section
is the BEC state by \eqref{eq:hubbard-phonon-bec-state-mixture}, and it remains to
prove cyclicity.
Let $\Psi$ be orthogonal to the cyclic subspace, $g\in\sphilbfrak{D}_{\txtbsn,1}$,
$z\in\fldcmp$, and $g_n=g+f_n$ with $f_n$ for $z-\fun{\faftr{g}}{0}$.
Then $\fun{\faftr{g_n}}{0}=z$, and
$\fun{\oarepn_{\txtbsn,\sminvtemperature}}{\fun{\opfockweyl}{g_n}}\oagnsvector[\Psi_{\txtbsn,\sminvtemperature}]\to\fun{\oarepn_{\txtbsn,\sminvtemperature}}{\fun{\opfockweyl}{g}}\oagnsvector[\Psi_{\txtbsn,\sminvtemperature}]$.
By the cyclicity of $\oagnsvector[\Psi_{\txtbsn,\sminvtemperature}]$, it follows that
$\int\napiernum^{\imunit\sqrt{2\fun{\rho_{\txtbsn,0}}{\sminvtemperature} r}\opreal\rbk{\napiernum^{\imunit\theta}z}}\bkt{\fun{\Psi}{r,\theta}}{\Phi}\opdmsr{\msrprb(r,\theta)}=0$
for all $z\in\fldcmp$ and $\Phi\in\sphilbfrak{H}_{\txtbsn,\sminvtemperature}$.
In the coordinates $\sqrt{r}\napiernum^{\imunit\theta}$ this is the Fourier
transform of a finite complex measure on $\fldreal^2$, and hence $\Psi=0$.
\end{proof}

Thus the GNS representation of the BEC state is the direct integral of the component representations over the order coordinate with the law \(\msrprb\), and distinct components have no nonzero intertwiner; the order coordinate is detected by the strong limit \eqref{eq:hubbard-phonon-bec-central-limit}.

\section{AI statement}\label{ai-statement}

Large language models were used to translate and reorganize the author's draft, to proofread the English text, to check cross-references and bibliographic citations, and to check the details of the proofs for errors and inconsistencies. The author reviewed and verified all suggested changes and remains fully responsible for the mathematical content and the text of this paper.

\bibliography{myref.bib}

\end{document}